\documentclass{birkjour_t2}

\usepackage{amsmath, amssymb}
\usepackage{amsthm , amsfonts, latexsym}

\usepackage[mathscr]{euscript}
\usepackage{algorithm,algorithmic}
\usepackage{graphics,graphicx}
\usepackage{enumerate}
\usepackage{cleveref}
\usepackage{color}
\newtheorem{thm}{Theorem}[section]
\newtheorem{cor}[thm]{Corollary}
\newtheorem{lem}[thm]{Lemma}

\theoremstyle{definition}
\newtheorem{defn}[thm]{Definition}
\theoremstyle{remark}

\numberwithin{equation}{section}
\begin{document}
\sloppy
%
%
%
%
%
%
%
%
%
\title[Spiders can be recognized by counting their legs]
 {Spiders can be recognized by counting their legs}

\author[Berkemer]{Sarah Berkemer}
\address{
Bioinformatics Group, \\ Department of Computer Science and
Interdisciplinary Center for Bioinformatics \\
University of Leipzig,\\
H{\"a}rtelstrasse 16-18, D-04107 Leipzig, Germany }
\email{bsarah@bioinf.uni-leipzig.de}

\author[Chaves]{Ricardo Chaves}
\address{
Departamento de Ci{\^e}ncia da Computa{\c{c}}{\~{a}}o (CIC)\\
Instituto de Ci{\^e}ncias Exatas\\
Universidade de Bras{\'\i}lia\\
Campus Universit{\'a}rio - Asa Norte \\
Bras{\'\i}lia DF - CEP: 70910-900\\
 Brasil
 }
\email{rrcchaves@gmail.com}

\author[Fritz]{Adrian Fritz}
\address{%
Center for Bioinformatics \\
Saarland University \\
Building E 2.1, Room 413 \\
P.O. Box 15 11 50 \\
D - 66041 Saarbr\"{u}cken \\
Germany }
\email{s9adfrit@stud.uni-saarland.de}

\author[Hellmuth]{Marc Hellmuth}
\address{%
Center for Bioinformatics \\
Saarland University \\
Building E 2.1, Room 413 \\
P.O. Box 15 11 50 \\
D - 66041 Saarbr\"{u}cken \\
Germany }
\email{marc@bioinf.uni-leipzig.de}

\author[H.-Rosales]{Maribel Hernandez-Rosales}
\address{
Departamento de Ci{\^e}ncia da Computa{\c{c}}{\~{a}}o (CIC)\\
Instituto de Ci{\^e}ncias Exatas\\
Universidade de Bras{\'\i}lia\\
Campus Universit{\'a}rio - Asa Norte \\
Bras{\'\i}lia DF - CEP: 70910-900\\
 Brasil
 }
\email{maribel@bioinf.uni-leipzig.de}

\author[Stadler]{Peter F. Stadler}
\address{
Bioinformatics Group, \\
Department of Computer Science; and
Interdisciplinary Center for Bioinformatics,\\
University of Leipzig,\\
H{\"a}rtelstrasse 16-18, D-04107 Leipzig,Germany\\[0.1cm]
Max Planck Institute for Mathematics in the Sciences\\
Inselstrasse 22, D-04103 Leipzig, Germany\\[0.1cm]
RNomics Group, Fraunhofer Institut f{\"u}r Zelltherapie und Immunologie,
Deutscher Platz 5e, D-04103 Leipzig, Germany\\[0.1cm]
Department of Theoretical Chemistry,  University of Vienna,
  W{\"a}hringerstra{\ss}e 17, A-1090 Wien, Austria\\[0.1cm]
Santa Fe Institute, 1399 Hyde Park Rd., Santa Fe, NM87501, USA}
\email{studla@bioinf.uni-leipzig.de}

\thanks{%
  This work was funded in part by the German Research Foundation (DFG)
  (Proj.\ No.\ MI439/14-1).  }



\date{\today}

\dedicatory{To all arachnophobic mathematicians.}

\begin{abstract}
  Spiders are arthropods that can be distinguished from their closest
  relatives, the insects, by counting their legs. Spiders have 8, insects
  just 6. Spider graphs are a very restricted class of graphs that
  naturally appear in the context of cograph editing. The vertex set of a
  spider (or its complement) is naturally partitioned into a clique (the
  body), an independent set (the legs), and a rest (serving as the
  head). Here we show that spiders can be recognized directly from their
  degree sequences through the number of their legs (vertices with degree
  $1$). Furthermore, we completely characterize the degree sequences of
  spiders.
\end{abstract}

\maketitle

\section{Introduction}
Determining the evolutionary history of species based on sequence
information is one of the main challenges in phylogenomics. Recent advances
in phylogenetics have shown that cographs play a crucial role in order to
clean up orthology, resp., paralogy data, i.e., estimates of pairs of genes
that arose from a speciation, resp., duplication event in the gene tree
\cite{HHH+13,HHH+14}. However, the (decision version of the) problem to
edit a given graph into a cograph is NP-complete \cite{Liu:11,Liu:12}.
Nevertheless, the cograph editing problem can be solved in polynomial time,
whenever the graph under investigation is so-called $P_4$-sparse
\cite{Liu:11, Liu:12}. The structure of $P_4$-sparse graphs can be
characterized in terms of so-called spiders \cite{Jamison:89,Jamison:92}.
Hence, the cograph editing problem can be rephrased as a ``spider editing
problem''. To address the spider editing problem in future work, we give in
this contribution a full characterization of spiders in terms of their
degrees and respective degree sequences.

\section{Preliminaries}
We consider only simple graphs $G=(V,E)$. The \emph{degree} $\deg_G(v)$ of
a vertex $v\in V$ is defined as the number of edges that contain $v$. If
there is no risk of confusion we write simply $\deg(v)$ instead of
$\deg_G(v)$. The \emph{degree sequence} $(n_0,n_1,n_2\dots,n_w)$ is a list
of non-negative integers. A degree sequence is \emph{graphical}, if it can
be realized by a graph, i.e., if there is a graph $G$ s.t.\ there are
exactly $n_k$ vertices with degree $k$ in $G$. Of course, for simple graphs
we have $n_k=0$ for $k\ge |V|$. It is well-known that one can verify in
polynomial time, whether a degree sequence is graphical, either by means of
the Havel-Hakimi algorithm \cite{Havel:55,Hakimi:62} or by using the
Erd{\H{o}}s-Gallai Theorem \cite{EG:60}.

A graph $H$ is a \emph{subgraph} of a graph $G$, in symbols $H\subseteq G$,
if $V(H)\subseteq V(G)$ and $E(H)\subseteq E(G)$. A subgraph $H \subseteq
G$ is called \emph{induced}, when $xy\in E(H)$ if and only if $xy\in E(G)$
for all $x,y\in V(H)$. The \emph{complement} $\overline{G}$ of a graph
$G=(V,E)$ has vertex set $V$ and edge set $\overline{E}=\{xy\mid xy\not\in
E \text{ for all distinct } x,y\in V\}$.

A graph $G$ is called \emph{$P_4$-sparse} if every set of five vertices
induces at most one $P_4$ (chordless path on four vertices)
\cite{Hoang:85}. The efficient recognition of $P_4$-sparse graphs is
intimately connected to spider graphs, as the next lemma shows.

\begin{lem}[\cite{Jamison:89,Jamison:92}.]
A graph $G$ is $P_4$ sparse if and only if exactly one of the following three
alternatives is true for every induced subgraph $H$ of $G$: (i) $H$ is not
connected, (ii) $\overline{H}$ is not connected, or (iii) $H$ is a spider
\end{lem}

Spiders come in two sub-types, called \emph{thin} and \emph{thick}.

\begin{defn} \cite{Jamison:92,Nastos:12} A graph $G$ is a \emph{thin
    spider} if its vertex set can be partitioned into three sets $K$, $S$,
  and $R$ so that (i) $K$ is a clique; (ii) $S$ is a stable set; (iii)
  $|K|=|S|\ge 2$; (iv) every vertex in $R$ is adjacent to all vertices of
  $K$ and none of the vertices of $S$; and (v) each vertex in $K$ has a
  unique neighbor in $S$ and \emph{vice versa}.  \newline A graph $G$ is a
  \emph{thick spider} if its complement $\overline{G}$ is thin spider.
\end{defn}
The sets $K$, $S$, and $R$ are usually referred to as the \emph{body}, the
set of \emph{legs}, and \emph{head}, resp., of a thin spider, see Fig.\
\ref{fig:3col}. The path $P_4$ is the only graph that is both a thin and
thick spider.

\begin{figure}[tbp]
  \centering
  \includegraphics[bb= 237 495 393 591, scale=1.0]{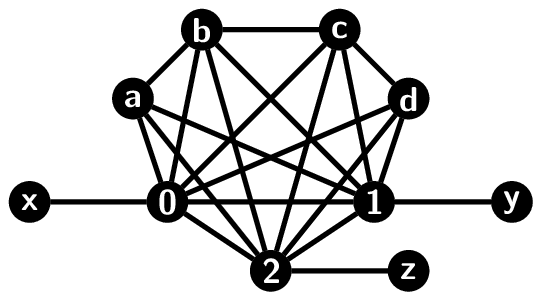}
  \caption{This thin spider, with vertex partition $K=\{0,1,2\}$,
    $S=\{x,y,z\}$ and $R=\{a,b,c,d\}$ has degree sequence
    $(0,3,0,0,2,2,0,3)$. In this example there is another spider hidden in
    the head $R$, namely $K'=\{b,c\}$, $S'=\{a,d\}$, and $R'=\emptyset$.}
  \label{fig:3col}
\end{figure}

\section{Results}

Spider graphs have characteristic degree sequences. Suppose $G$ is a thin
spider. It follows directly from the definition that $\deg(x)=1$ for all
$x\in S$: The leg $x$ is adjacent to a single body vertex and to none of
the head vertices (Condition (iv) and (v)). Each body vertex $x\in K$ has
$\deg(x)=|K|+|R|$, since it is adjacent to $1$ leg, all $|K|-1$ other body
vertices, and all $|R|$ head vertices. For a head vertex $x\in R$ we have
$|K|\le \deg(x)\le |K|+|R|-1$ since it is adjacent at least to all body
vertices and at most to all body and head vertices except itself. As
$|K|+|R|=|V|-|S|$ we can rephrase these observations in the following form,
which have in parts already be discovered in \cite[Obs.\ 2.8 \&
2.9]{Jamison:92}.
	
\begin{lem}
  \label{lem:thinspider}
  If $G$ is a thin spider, then it has $|S|\ge 2$ legs with degree $1$,
  $|S|$ body vertices with degree $|V|-|S|$, and $|V|-2|S|$ head vertices
  $x\in R$ with $2\le |S|\le \deg(x)\le |V|-|S|-1$.
\end{lem}

Our main result is that this condition is also sufficient to identify thin
spiders:
\begin{thm} 
  A graph $G=(V,E)$ with vertex set $V$ is a thin spider if and only if
  there is an integer $s\ge 2$ so that $G$ has exactly $s$ vertices of
  degree $1$ and exactly $s$ vertices of degree $|V|-s$.
\label{thm:1}
\end{thm}
\begin{proof}
  The ``only if'' part of this statement is a consequence of
  Lemma~\ref{lem:thinspider}.
  
  \noindent
  Hence, consider a graph $G$ with $s\ge 2$ nodes with degree $1$ and $s$
  nodes of degree $|V|-s$. Define $K:=\{x\in V| \deg(x)=|V|-s\}$ as the set
  of body and $S:=\{x\in V| \deg(x)=1\}$ as the legs.

  \smallskip
  \noindent  
  \textbf{Claim 1.} \emph{For each body vertex $x\in K$ there is a unique leg
  $y\in S$ so that $xy\in E$.}

  \noindent	
  If $x$ is not adjacent to a leg $y\in S$ then it has at most
  $|V|-|S|-1$ neighbors, a contradiction. Thus, every body vertex $x\in K$
  has at least one adjacent leg $y\in S$. Since $|K|=|S|$ by assumption,
  every body node is adjacent to exactly one leg.  \hfill{$\triangleleft$}
  \smallskip

  \noindent
  This establishes property (v).  \smallskip

  \noindent
  \textbf{Claim 2.} \emph{Every body vertex $x\in K$ is adjacent to all
    non-leg vertices $y\in V\setminus S$, except itself.}

  \noindent
  Since $\deg(x)=|V|-|S|$ and one of the neighbors of $x$ is a leg, it is
  adjacent to $|V|-|S|-1$ non-leg vertices. There are in total $|V|-|S|$
  non-leg vertices, i.e., $|V|-|S|-1$ non-leg vertices other than $x$
  itself.  Since $x$ cannot be adjacent to itself, it is adjacent to all
  other non-leg vertices.  \hfill{$\triangleleft$} 
  \smallskip

  \noindent
  It follows that $K$ forms a clique in $G$, i.e., property (i) holds.
  Furthermore every head node $z\in R:=V\setminus(K\cup S)$ is connected to
  every body node and to none of the legs, i.e., property (iv)
  holds. Conditions (ii) and (iii) are satisfied by construction, thus $G$
  is a thin spider.
\end{proof}

Since the complement of a thick spider is a thin spider and \emph{vice
  versa}, the characterization of thin spiders immediately implies a
characterization of thick spiders because
$\deg_{\overline{G}}(x)+\deg_G(x)=|V|-1$. It follows that a thick spider
has $|S|$ body vertices of degree $|V|-2$ and $|S|$ leg vertices of degree
$|S|-1$.

\begin{cor} 
  $G$ is a thick spider if and only if there is an integer $s\ge 2$ so that
  $G$ has exactly $s$ vertices with degree $s-1$ and $s$ vertices with 
  degree $|V|-2$.
\end{cor}

\begin{thm} Let  $\Pi = (n_0,n_1,\dots,n_w)$ be an arbitrary degree sequence. 
  Then it holds that $\Pi$ is realizable by a thin spider if and only if
  there are integers $s\ge 2$, $v\ge 4$, $s<v$ so that
  \begin{itemize}
  \item[(i)] $n_1=s$, $n_{v-s}=s$, and
  \item[(ii)] $\Pi' = (m_0,m_1,\dots,m_w)$ is a graphical degree sequence
    where
    \begin{equation*}
      m_k= 
      \begin{cases}
        n_k = 0, & \text{if }  2\le k<s \text{ or } k\geq v-s\\
        n_k-s,   & \text{otherwise}.
      \end{cases}
    \end{equation*}
  \end{itemize}

  A sequence $(n_1,\dots,n_j,\dots, n_{|V|-1})$ is the degree sequence of a
  thick spider if and only if $(n_{|V|-1},\dots,n_{|V|-1-j},\dots, n_1)$ is
  the degree sequence of a thin spider.
\end{thm}
\begin{proof}
  Let $\Pi$ be realizable by a thin spider. Theorem \ref{thm:1} implies
  Conditions (i). As already observed, the $|V|-2|S|$ head vertices of a
  thin spider have degrees in the range from $|S|$ to $|V|-|S|-1$, that is
  reflected by the given degrees of $\Pi'$ in Condition (ii). Moreover,
  since $\Pi$ is the degree sequence of a thin spider $G$, $R$ induces a
  subgraph of $G$ and each vertex in $R$ has $s$ neighbors in $K$, we can
  conclude that $\Pi'$ is a graphical degree sequence. Note, from the
  algorithmic point of view, Condition (ii) can easily be checked by the
  Erd{\H{o}}s-Gallai Theorem \cite{EG:60}.

  Now assume that there are integers $s\ge 2$, $v\ge 4$, $s<v$ so that
  Conditions (i) and (ii) are fulfilled for $\Pi$. We construct a spider
  from the graph $(V,E)$ with $|V|=v$ and $E=\emptyset$. Choose a subset
  $K\subseteq V$ with $|K|=s$, $S\subseteq V\setminus K$ with $|S|=s$, and
  set $R=V\setminus (K\cup S)$. Now, add edge $xy$ to $E$ for all distinct
  $x,y\in K$, choose a bijection $f:K\to S$ and add $xf(x)$ to $E$ for all
  $x\in K$ and finally add edge $ab$ to $E$ for all $a\in K,b\in R$.
  Clearly, the construction is feasible in terms of Condition (i) and (ii),
  and fulfills the properties of a thin spider. Finally, Condition (ii)
  also ensures that $R$ induces itself a graph, that can be constructed
  with the Havel-Hakimi algorithm \cite{Havel:55, Hakimi:62}.

  The result for thick spiders follows by observing that the degree
  sequence of $\overline{G}$ is related to the degree sequence of $G$ via
  $\bar n_{k}= n_{|V|-1-k}$ for $0\le k< |V|-1$.
\end{proof}

\bibliographystyle{plain}
\bibliography{cogra}

\end{document}